\documentclass{amsart}
\pdfoutput=1

\usepackage{amssymb}
\usepackage{amsthm}
\usepackage{amsfonts}
\usepackage{amsmath}
\usepackage{pmboxdraw}
\usepackage{verbatim} 
\usepackage{graphicx}
\usepackage{color}
\usepackage[colorlinks=true, citecolor=blue, filecolor=black, linkcolor=black, urlcolor=black]{hyperref}
\usepackage{cite}
\usepackage[normalem]{ulem}
\usepackage{subcaption}
\usepackage{bbm}
\usepackage{bm}
\usepackage{mathtools}
\usepackage{todonotes}
\usepackage{kantlipsum}
\allowdisplaybreaks

\newcommand{\RN}[1]{%
	\textup{\uppercase\expandafter{\romannumeral#1}}%
}

\def\pa{\partial}

\def\R{\mathbb{R}}

\newcommand{\erfc}{\operatorname{erfc}}
\newcommand{\erf}{\operatorname{erf}}

\newcommand{\re}{\operatorname{Re}}

\newcommand{\Tr}{\operatorname{Tr}}

\theoremstyle{plain}
\newtheorem{thm}{Theorem}[section]

\newtheorem{cor}[thm]{Corollary}

\newtheorem{prop}[thm]{Proposition}
\newtheorem{rem}[thm]{Remark}

\theoremstyle{remark}

\numberwithin{equation}{section}

\ExplSyntaxOn

\NewDocumentCommand{\pFq}{O{}mmmmm}
 {
  \group_begin:
  \keys_set:nn { hypergeometric } { #1 }
  \hypergeometric_print:nnnnn { #2 } { #3 } { #4 } { #5 } { #6 }
  \group_end:
 }
\NewDocumentCommand{\hypergeometricsetup}{m}
 {
  \keys_set:nn { hypergeometric } { #1 }
 }

\tl_new:N \l_hypergeometric_divider_tl
\tl_new:N \l_hypergeometric_left_tl
\tl_new:N \l_hypergeometric_right_tl

\keys_define:nn { hypergeometric }
 {
  symbol .tl_set:N = \l_hypergeometric_symbol_tl,
  symbol .initial:n = F,
  separator .tl_set:N = \l_hypergeometric_separator_tl,
  separator .initial:n = {},
  skip .tl_set:N = \l_hypergeometric_skip_tl,
  skip .initial:n = 8,
  divider .choice:,
  divider/semicolon .code:n = \tl_set:Nn \l_hypergeometric_divider_tl { \;; },
  divider/bar .code:n = \tl_set:Nn \l_hypergeometric_divider_tl { \;\middle|\; },
  divider .initial:n = semicolon,
  fences .choice:,
  fences/brack .code:n = 
   \tl_set:Nn \l_hypergeometric_left_tl {(}
   \tl_set:Nn \l_hypergeometric_right_tl {)},
  fences/parens .code:n = 
   \tl_set:Nn \l_hypergeometric_left_tl {(}
   \tl_set:Nn \l_hypergeometric_right_tl {)},
  fences .initial:n = brack,
 }

\cs_new_protected:Nn \hypergeometric_print:nnnnn
 {
  {} \sb {#1} \l_hypergeometric_symbol_tl \sb { #2 }
  \left\l_hypergeometric_left_tl
  \genfrac .. 
           {0pt} 
           {} 
           { \__hypergeometric_process:n { #3 } } 
           { \__hypergeometric_process:n { #4 } } 
  \l_hypergeometric_divider_tl
  #5
  \right\l_hypergeometric_right_tl
 }

\cs_new_protected:Nn \__hypergeometric_process:n
 {
  \clist_use:nn { #1 }
   {
    {\l_hypergeometric_separator_tl}
    \mspace { \l_hypergeometric_skip_tl mu }
   }
 }

\ExplSyntaxOff

\hypergeometricsetup{
  fences=parens,
  separator={,},
  divider=bar,
}

\begin{document}

\title[Spectral moments of the real Ginibre ensemble]{Spectral moments of the real Ginibre ensemble}

\author{Sung-Soo Byun}
\address{Department of Mathematical Sciences and Research Institute of Mathematics, Seoul National University, Seoul 151-747, Republic of Korea}
\email{sungsoobyun@snu.ac.kr}

\author{Peter J. Forrester}
\address{School of Mathematical and Statistics, The University of Melbourne, Victoria 3010, Australia}
\email{pjforr@unimelb.edu.au} 

\thanks{Sung-Soo Byun was partially supported by the POSCO TJ Park Foundation (POSCO Science Fellowship) and by the New Faculty Startup Fund at Seoul National University.  Funding support to Peter Forrester for this research was through the Australian Research Council Discovery Project grant DP210102887.
}

\subjclass[2020]{Primary 60B20; Secondary 33C20}
\keywords{Real Ginibre ensemble, real eigenvalues, spectral moments, hypergeometric functions, recurrence relation, asymptotic expansions}

\begin{abstract}
The moments of the real eigenvalues of real Ginibre matrices are investigated from the viewpoint of explicit formulas, differential and difference equations, and large $N$ expansions. These topics are inter-related. For example, a third order differential equation can be derived for the density of the real eigenvalues, and this can be used to deduce a second order difference equation for the general complex moments $M_{2p}^{\rm r}$. The latter are expressed in terms of the ${}_3 F_2$ hypergeometric functions, with a simplification to the ${}_2 F_1$ hypergeometric function possible for $p=0$ and $p=1$, allowing for the large $N$ expansion of these moments to be obtained. The large $N$ expansion involves both integer and half integer powers of $1/N$. The three   term recurrence then provides the large $N$ expansion of the full sequence $\{ M_{2p}^{\rm r} \}_{p=0}^\infty$. Fourth and third order linear differential equations are obtained for the moment generating function and for the Stieltjes transform of the real density, respectively, and properties of  the large $N$ expansion of these quantities are determined.
\end{abstract}

\maketitle

\section{Introduction}

\subsection{Context}
The study of moments of the spectral density for random matrix ensembles hold a special place in the development of random matrix theory. One landmark was Wigner's \cite{Wi55,Wi58} introduction of what is now referred to as the moment method as a strategy to prove that a large class of symmetric random matrices have the same scaled spectral density $\rho^{\rm W}(x) := {2 \over \pi} (1 - x^2) \chi_{|x|<1}$. 
The functional form $\rho^{\rm W}(x)$ is now known as the Wigner semi-circle. An essential point is that the $2p$-th even moment $m_{2p}$ of $\rho^{\rm W}(x)$ is equal to $2^{-2p} C_p$, where $C_p$ denotes the $p$-th Catalan number, and moreover this moment sequence uniquely determines $\rho^{\rm W}(x)$. Wigner then used the fact that in the random matrix setting $m_{2p} = \lim_{N \to \infty} (\sigma^2/2^2N)^p\mathbb E ({\rm Tr} \, X^{2p})$ as the starting point for the computation of limiting scaled moments of the spectral density (here $\sigma^2$ is the variance of the off diagonal entries of $X$, which are assumed too to have mean zero). Subsequent to Wigner's work, the method of moments, and its generalisation to cumulants, has been used extensively in questions beyond the spectral density such as for the Gaussian fluctuations of linear statistics and pair counting statistics; see the recent review \cite{SW23}.

Specialise now to complex Hermitian matrices from the Gaussian unitary ensemble, with unit variance of the modulus of the off-diagonal entries so that the joint element distribution is proportional to $e^{-{\rm Tr} \, X^2/2}$. Another landmark has been in relation to the interpretation of the coefficients in the large $N$ of $M_{2p}^{\rm GUE} := \mathbb E \, {\rm Tr} \, X^{2p}$. This expansion is a terminating series in $1/N^2$,
\begin{equation} \label{1.0b}
N^{-p-1} M_{2p}^{\rm GUE} = \sum_{g=0}^{\lfloor p/2 \rfloor}
{\nu_{p,g}^{\rm GUE} \over N^{2g}}.
\end{equation}
The result of Wigner gives that $\nu_{p,0} = C_p$.
Using a diagrammatic interpretation of non-zero terms in the computation of ${\rm Tr} \, X^{2p}$ as implied by Wick's theorem, it was shown by Br\'ezin et al.~\cite{BIPZ78} that the $\nu_{p,g}$ count the number of pairings of the sides of a $2p$-gon which are dual to a map on a compact orientable Riemann surface of genus $g$. Equivalently, after the sides are identified in pairs, it is required that a surface of genus $g$ results. The case $g=0$ is referred to as planar, and represent the leading order in (\ref{1.0b}).

Motivated by this topological interpretation, Harer and Zagier \cite{HZ86} further investigated the sequences $\{ M_{2p}^{\rm GUE} \}$ and $\{ \nu_{p,g}^{\rm GUE} \}$.\footnote{The notation $C(p,N)$ in \cite{HZ86} is equivalent to $M_{2p}^{\rm GUE}$. } In particular they obtained that $\{ M_{2p}^{\rm GUE} \}$ obey the three term recurrence
\begin{equation} \label{3T}
(p+1) M_{2p}^{\rm GUE} = (4p - 2) N M_{2p-2}^{\rm GUE} + (p-1) (2p-1) (2p - 3) M_{2p-4}^{\rm GUE}
\end{equation}
subject to the initial conditions $M_{0}^{\rm GUE}= N$, $M_{2}^{\rm GUE} = N^2$. From this it was shown that $\{ \nu_{p,g} \}$ obey the two-variable recurrence
\begin{equation} \label{2T}
(p+2)  \nu_{p+1,g}^{\rm GUE} = p(2p+1)  (2p-1)  \nu_{p-1,g-1}^{\rm GUE} +2 (2p+1)  \nu_{p,g}^{\rm GUE}, 
\end{equation}
subject to the initial condition $\nu_{0,0}^{\rm GUE} = 1$, and boundary conditions $\nu_{p,g}^{\rm GUE} = 0$ for any of the conditions $k<0, g<0$ or $g > \lfloor p/2 \rfloor$.
For instance, the first few values are given by 
\begin{align*}
\nu_{p,0}^{ \rm GUE } = C_p , 
\qquad 
\nu_{p,1}^{ \rm GUE }  = C_p \, \frac{(p+1)!}{(p-2)!}  \frac{ 1}{12},  
\qquad  \nu_{p,2}^{ \rm GUE } = C_p \, \frac{(p+1)!}{(p-4)!} \frac{ 5p-2 }{1440},
\end{align*}
see e.g. \cite[Theorem 7]{WF14}.

The focus of our study relating to random matrix spectral moments in the present work is an outgrowth of theory underlying and relating the the three term recurrence (\ref{3T}), combined with results from the recent paper \cite{By23b} by one of us.
The question addressed in \cite{By23b} is to identify a recurrence relation for the spectral moments of the real eigenvalues of elliptic GinOE matrices. The latter is the ensemble of asymmetric real Gaussian matrices defined by
\begin{equation} \label{SS}
\sqrt{1 + \tau \over 2} S_+ +
\sqrt{1 - \tau \over 2} S_-,
\end{equation}
where $S_{\pm}$ are independent random real symmetric and skew-symmetric GOE matrices, and $0 \le \tau \le 1$ is a parameter. For $\tau = 1$ one sees that elliptic GinOE reduces to GOE. Earlier, the work of Ledoux \cite{Le09} had found a fifth order linear recurrence for $\{ M_{2p}^{\rm GOE} \}$. Existing literature due to Goulden and Jackson \cite{GJ97}, extending the work of Harer and Zagier, has given an interpretation of the $M_{2p}^{\rm GOE}$ in terms of pairings which lead to both nonoriented and orientable surfaces. 
The analogue of (\ref{1.0b}) is now
\begin{equation} \label{1.0b+}
N^{-p-1} M_{2p}^{\rm GOE} = \sum_{l=0}^{p}
{\nu_{p,l}^{\rm GOE} \over N^{l}},
\end{equation}
which in particular involves both odd an even powers of $1/N$. In keeping with the universality of the Wigner semi-circle law, one again has for the leading contribution $\nu_{p,0}^{\rm GOE}
= C_p$.

The other extreme of (\ref{SS}) is $\tau = 0$, when each entry is identically distributed as an independent standard real Gaussian, this giving rise to a random matrix from GinOE. See \cite{BF22,BF23} for recent reviews on the Ginibre ensembles. Here there was no previous literature on the moment sequence of the real eigenvalues. The $\tau = 0$ limiting case of the in general 11-term linear recurrence found in \cite{By23b} for the moments of the density of real eigenvalues was found to reduce to just a three term recurrence 
\begin{equation} \label{1.17}
2 (2p+1) M_{2p}^{\rm r, GinOE} =
(2p-1) (6p+4N - 5) M_{2p-2}^{\rm r, GinOE} 
- (2p-3) (2p+N - 4) (2p+2N - 3) M_{2p-4}^{\rm r, GinOE}.
\end{equation}

Motivated by the relative simplicity of (\ref{1.17}), and its similarity with the GUE moment recurrence (\ref{3T}), in this work we will carry out a study of the moments $\{ M_{2p}^{\rm r, GinOE} \}$ as a stand alone sequence, not viewed as a limit of 
moments for the real eigenvalues of elliptic GinOE. In the case of the recurrence (\ref{3T}) it has been known since the work of Haagerup and Thorbjørnsen \cite{HT03} that there is a tie in with certain higher order differential equation and also with certain special function functions, in particular hypergeometric polynomials. In fact such structures have been shown to also be features of the spectral moments in a broad range of settings \cite{Le09,MS11,CMOS19,CCO20,RF21,ABGS21,FR21,Fo21,FLSY23,GGR21}. However GinOE is distinct from the ensembles in these earlier studies since only a fluctuating fraction of eigenvalues are real.

\subsection{Some known results}

Let $G$ be a real Ginibre matrix (GinOE) of size $N$, defined by the requirement that all entries are independent standard Gaussians.  
By making use of knowledge of the Schur function average with respect to GinOE matrices, it was shown in \cite{SK09,FR09} that for any positive integer $p \ge 1$, 
\begin{equation} \label{spectral moments for real and complex}
\mathbb{E} \Big[ \Tr G^{2p} \Big] = 2^p \frac{ \Gamma(N/2+p) }{ \Gamma(N/2) } = N(N+2)\dots (N+2p-2). 
\end{equation}
But with the eigenvalues of GinOE matrices being in general both real and complex, this is a result which combines moments relating to the real eigenvalues, and moments relating to the complex eigenvalues.
Specifically,
let $\mathcal{N}_\R$ be the number of real eigenvalues and define
\begin{equation}
M_{2p,N}^{ \rm r } := \mathbb{E} \bigg[ \sum_{j=1}^{ \mathcal{N}_\R } x_j^{2p}  \bigg], \qquad M_{2p,N}^{ \rm c } := \mathbb{E} \bigg[ \sum_{j=1}^{ N-\mathcal{N}_\R } z_j^{2p} \bigg] = 2 \, \mathbb{E} \bigg[ \sum_{j=1}^{ (N-\mathcal{N}_\R)/2 } \re z_j^{2p} \bigg], 
\end{equation}
where $x_j$ $(j=1,\dots, \mathcal{N}_\R)$ and $z_j$ ($j=1,\dots, N-\mathcal{N}_\R$) are real and complex eigenvalues of $G$, respectively. 
Here and in the sequel, we drop the superscript ``GinOE'', cf. \eqref{1.17}. 
We have used the convention $z_{j+(N-\mathcal{N}_\R)/2}=\bar{z}_j.$    
Note that by definition, 
\begin{equation} \label{real+cplx}
M_{2p,N}^{ \rm r }+ M_{2p,N}^{ \rm c } = \mathbb{E} \Big[ \Tr G^{2p} \Big] . 
\end{equation}

It has been  known for some time \cite{EKS94,Ed97} that the average densities of real and complex eigenvalues of the GinOE are given by 
\begin{align}
\rho_N^{ \rm r }(x) &=  \frac{1}{\sqrt{2\pi}(N-2)! } \bigg(  \Gamma(N-1,x^2) +  2^{(N-3)/2} e^{ -\frac{x^2 }{2} } |x|^{N-1} \gamma\Big( \frac{N-1}{2}, \frac{x^2}{2} \Big)  \bigg),  \label{real density}
\\
\rho_N^{ \rm c }(x+iy) & = \sqrt{ \frac{2}{\pi}  } |y| \erfc(\sqrt{2}|y|) e^{2y^2} \frac{ \Gamma(N-1,x^2+y^2) }{ \Gamma(N-1) }. \label{complex density}
\end{align}
Here, 
$$
\gamma(a,z) := \int_0^z e^{a-1} e^{-t}\,dt, \qquad \Gamma(a,z) :=  \int_z^\infty e^{a-1} e^{-t}\,dt =\Gamma(a)- \gamma(a,z)
$$
are lower and upper incomplete gamma functions and 
$$
\erfc(z) := \frac{2}{ \sqrt{\pi} } \int_z^\infty e^{-t^2} \,dt
$$
is the complementary error function. 
The densities relate to the even integer moments of the real and complex eigenvalues by
\begin{equation}\label{1.11}
M_{2p,N}^{ \rm r } = \int_\R x^{2p} \rho_N^{ \rm r }(x) \,dx ,\qquad M_{2p,N}^{ \rm c } = \int_{\R^2} (x+iy)^{2p} \rho_N^{ \rm c }(x+iy) \,dx\,dy.
\end{equation}

The pioneering work \cite{EKS94} relating to the real eigenvalues of GinOE has provided both an exact evaluation, and an asymptotic expansion, for
$$
M_{0,N}^{ \rm r } = N- M_{0,N}^{ \rm c }  = \mathbb{E} \mathcal{N}_\R.
$$
Thus, from \cite[Cor.~5.1]{EKS94} we have the expression in terms of a particular Gauss hypergeometric function
\begin{equation}\label{1.12}
M_{0,N}^{ \rm r } = 
\frac12 + \sqrt{ \frac{2}{\pi} } \frac{ \Gamma(N+\frac12) }{ (N-1)! } \pFq{2}{1}{1,-\frac12}{N}{\frac12}.
\end{equation}
As an application of this formula, it is shown in \cite[Cor.~5.2]{EKS94} that for $N \to \infty$
\begin{equation}\label{1.13}
M_{0,N}^{ \rm r } =
\sqrt{2N \over \pi} \bigg ( 1 -
{3 \over 8 N} - {3 \over 128 N^2} + {27 \over 1024 N^3} + {499 \over 32768 N^4} + {\rm O} \Big ( {1 \over N^5} \Big ) \bigg ) + {1 \over 2}.
\end{equation}


As a series, the Gauss hypergeometric function in (\ref{1.12}) is not terminating. Nonetheless, by considering the recurrence in $N$ implied by this formula, such terminating forms were obtained \cite[Cor.~5.3]{EKS94}
\begin{equation}
M_{0,N}^{ \rm r } =   \begin{cases}
\displaystyle   1+ \sqrt{2} \sum_{k=1}^{(N-1)/2} \frac{(4k-3)!!}{ (4k-2)!! } &N \textup{ odd},
 \smallskip
   \\
\displaystyle     \sqrt{2} \sum_{k=0}^{N/2-1} \frac{(4k-1)!!}{ (4k)!! } &N \textup{ even}.
\end{cases}
\end{equation}
(See also \cite[Prop.~2.1]{ABES23}.)
As a consequence, one reads off that for $N$ even $M_{0,N}^{ \rm r }$ is equal to $\sqrt{2}$ times a rational number, while for $N$ odd it is equal to 1 plus $\sqrt{2}$ times a rational number. As an aside, we remark that an arithmetic result of this type is also known for the expected number of real eigenvalues in the case of the product of two size $N$ GinOE matrices, where it is shown in \cite[\S 4.2]{FI16} to be of the form $\pi$ times a rational number for $N$ even, and 1 plus $\pi$ times a rational number for $N$ odd. The special function that appears here is not the Gauss hypergeometric function but rather a particular Meijer G-function, which was simplified to a finite series by Kumar \cite{Ku15}.

\subsection{New results}
Our first new result generalises (\ref{1.13}).

\begin{thm} \label{Thm_large N}
Let $m$ be any positive integer. 
Then as $N \to \infty$, we have 
\begin{equation} \label{M0N asymp}
M_{0,N}^{ \rm r } =  \sqrt{ \frac{2}{\pi} N } \bigg( 1+ \sum_{ l=1 }^{m-1} \frac{ a_l }{ N^l } +{\rm O}( N^{-m} ) \bigg)+\frac12, 
\end{equation}
where 
\begin{equation}
 a_l= -  \frac{1}{ \sqrt{\pi} } \frac{ \Gamma(l-\frac12) }{ l! } \frac{d^l }{dt^l} \bigg[ \Big( \frac{e^t-1}{t} \Big)^{-\frac32}  \frac{  e^{ 2t  } }{ e^t+1 }   \bigg]_{t=0}.
\end{equation}
Furthermore, as $N \to \infty$, we have 
\begin{equation}  \label{M2N asymp}
N^{-1} M_{2,N}^{ \rm r } =  \sqrt{ \frac{2}{\pi} N } \bigg( \frac{1}{3}+ \sum_{l=1}^{m-1} \frac{b_l}{N^l}  +{\rm O}( N^{-m} ) \bigg) + \frac{1}{2},
\end{equation} 
where 
\begin{equation}
b_l=-\frac{1}{2\sqrt{\pi} } \frac{ \Gamma(l-\frac32) }{  l! }  \frac{ d^l }{ dt^l } \bigg[ \Big( \frac{e^t-1}{t} \Big)^{-\frac{5}{2} }  \frac{ e^{2t}(e^t-3) }{ (e^t+1)^2 }   \bigg]_{t=0}. 
\end{equation} 
This gives for the first few terms
\begin{equation}
N^{-1} M_{2,N}^{ \rm r } =  \sqrt{ \frac{2}{\pi} N } \bigg( \frac13 + \frac{3}{8N} -\frac{43}{384N^2} + \frac{29}{1024N^3} + \frac{1859}{ 98304N^4 } +{\rm O} \Big (\frac{1}{N^5} \Big )  \bigg)+\frac12. 
\end{equation}
The three term recurrence (\ref{1.17}) then gives that for all $p \ge 0$, 
\begin{equation}\label{1.21}
N^{-p} M_{2p,N}^{ \rm r } = \sqrt{ \frac{2}{\pi} N } \bigg( {1 \over 2p + 1} + \sum_{l=1}^{m-1} {b_{l,p} \over N^l} + {\rm O}(N^{-m}) \bigg ) + {1 \over 2} + \sum_{l=1}^{p-1} {c_{l,p} \over N^l} 
\end{equation}
for certain coefficients $b_{l,p}$ and $c_{l,p}.$
\end{thm}

We remark that a generalisation of the asymptotic formula \eqref{M0N asymp} for the elliptic GinOE can also be found in \cite[Prop.~2.2]{BKLL23}. 
The terminating series $ \sum_{l=1}^{p-1} c_{l,p}  N^{-l} $ for the first $p=2,3,4,5$ are given by 
\begin{align} \label{terminating p2345}
\frac{1}{N}, \qquad \frac{3}{N} + \frac{4}{N^2}, \qquad \frac{6}{N} + \frac{22}{N^2} + \frac{24}{N^3}, \qquad \frac{10}{N} +\frac{70}{N^2} + \frac{200}{N^3} + \frac{192}{N^4},  
\end{align}
respectively. These coefficients can also be derived from the moment generating function, see \eqref{u 32 0} and \eqref{u 52 0}. 

\medskip 

Recall that the generalised hypergeometric function is given by the Gauss series
\begin{equation} \label{def of gen hypergeometric}
 \pFq{r}{s}{c_1,\dots,c_r}{d_1, \dots,d_s}{z} := \sum_{k=0}^\infty \frac{(c_1)_k \dots (c_r)_k }{ (d_1)_k \dots (d_s)_k } \frac{z^k}{k!};
\end{equation}
see e.g. \cite[Chapter 16]{NIST}, where it is assumed that the parameters are such that the series converges.
Using this function with $r=3$ and $s=2$ we next give an explicit formula for the even $2p$-th moments of both the real and complex eigenvalue density.
In earlier work the hypergeometric function ${}_3F_2$ has appeared in the evaluation of spectral moments of certain Hermitian unitary ensembles --- see \cite[Eqs.(3.9),(3.10)]{CMOS19} --- although there with extra structure of being terminating and furthermore identifiable in terms of certain orthogonal polynomials in the Askey scheme.

\begin{thm} \label{Thm_GinOE moments}
For all positive integers $N$ and $p$ we have 

    \begin{align}
\begin{split}  \label{M 2pN r}
M_{2p,N}^{ \rm r }  
&= \frac{1}{ \sqrt{2\pi} } \frac{ 2 }{ 2p+1 }  \frac{ \Gamma(N+p-\frac12) }{ (N-2)! }  \pFq{3}{2}{1,-\frac12-p,\frac12+p}{ \frac12, \frac32-N-p }{\frac12}  +  2^p \frac{ \Gamma(p+N/2)  }{ \Gamma(N/2) } \mathbbm{1}_{ \{ N: \textup{ odd} \} }
\end{split}
\end{align} 
and 
\begin{align}
\begin{split}  \label{M 2pN c}
M_{2p,N}^{ \rm c }  
&= -\frac{1}{ \sqrt{2\pi} } \frac{ 2 }{ 2p+1 }  \frac{ \Gamma(N+p-\frac12) }{ (N-2)! }  \pFq{3}{2}{1,-\frac12-p,\frac12+p}{ \frac12, \frac32-N-p }{\frac12}  +  2^p \frac{ \Gamma(p+N/2)  }{ \Gamma(N/2) } \mathbbm{1}_{ \{ N: \textup{ even} \} }. 
\end{split}
\end{align}

\end{thm}
The definition (\ref{1.11}) of the even integer moment $M_{2p,N}^{\rm r}$ can be extended to all complex $p$ with Re$(p) > - 1/2$ by rewriting $M_{2p,N}^{\rm r}$ as 
\begin{equation}\label{1.11a}
M_{2p,N}^{\rm r} = \int_{\mathbb R} |x|^{2p} \rho_N^{\rm r}(x) \, dx = 2 \int_0^\infty x^{2p} \rho_N^{\rm r}(x) \, dx.
\end{equation}
The evaluation formula (\ref{M 2pN r}) again remains valid.

\begin{prop}\label{P1.3}
Define $M_{2p,N}^{\rm r}$ for general $\re (p) > - 1/2$ by (\ref{1.11a}). 
The evaluation formula (\ref{M 2pN r}) can be continued off the positive integers to remain valid throughout this region of the complex plane.
\end{prop}

\begin{rem}
Let $p=q+1/2$ for $q \ge 0$ a non-negative integer. The series 
(\ref{def of gen hypergeometric}) defining the ${}_3F_2$ function in (\ref{M 2pN r}) is ill-defined as the parameters are such that the indeterminant zero divided by zero occurs. For the series to be well defined, the limit $q$ approaches a non-negative integer must be taken.
\end{rem}

It is possible to deduce that the three term recurrence (\ref{1.17}) for the moments is valid not just for the even integer moments, but the complex moments too, and to use this to deduce a three term recurrence specifically for the ${}_3F_2$ function appearing in (\ref{M 2pN r}). To this end, we make use of a  third order differential equation satisfied by $\rho_N^{\rm r}(x)$, which is of independent interest.

\begin{prop}\label{P1.5}
The density $\rho_N^{ \rm r }$ of real eigenvalues given in \eqref{real density} satisfies the differential equation
\begin{equation} \label{DE of real density}
\mathcal A_N[x] \rho_N^{ \rm r } (x)=0, \quad
\mathcal A_N[x] :=
\Big( x^2 \partial_x^3 + x ( 3 x^2-3N+4) \partial_x^2 + ( 2 x^2-2N+1 ) (  x^2-N+2 ) \partial_x \Big).
\end{equation}
\end{prop}

\begin{cor}\label{C1.6}
The three term recurrence (\ref{1.17}) remains valid for complex values of $p$ such the terms are well defined. Also, as a function of complex $p$
\begin{align}\label{1.28}
\begin{split}
 & (2N+2p-1)  (2N+2p+1)    \pFq{3}{2}{1,-\frac52-p,\frac52+p}{ \frac12, -\frac12-N-p }{\frac12}  
 \\
&\qquad =  (6p+4N+7) (2N+2p-1)      \pFq{3}{2}{1,-\frac32-p,\frac32+p}{ \frac12, \frac12-N-p }{\frac12} 
\\
&\qquad \quad - 2 (2p+N)(2N+2p+1)   \pFq{3}{2}{1,-\frac12-p,\frac12+p}{ \frac12, \frac32-N-p }{\frac12} . 
\end{split}
\end{align}

\end{cor}

Another consequence of \eqref{DE of real density} is related to the differential equation for the Fourier-Laplace, or equivalently for the (positive integer) moment generating function, as well as for the Stieltjes transform.

\begin{cor} \label{Cor_DE of MGF}
Let
\begin{equation}
u(t) := \int_\R e^{tx} \rho_N^{ \rm r }(x)\,dx. 
\end{equation} 
This satisfies the fourth order linear differential equation
\begin{equation} \label{DE MGF}
D_N[t]\, u(t)=0,  
\end{equation}
where 
\begin{equation}
D_N[t]:= 2t \, \partial_t^4 - ( 3 t^2-8 ) \, \partial_t^3 + t(t^2 - 4 N   -13 ) \, \partial_t^2 + (  (3 N+2) t^2 -8N-8  )\, \partial_t +(2N^2+N)t.  
\end{equation} 
Furthermore, introduce the Stieltjes transform of the density
\begin{equation}
W(t) := \int_\R { \rho_N^{ \rm r }(x) \over t - x} \,dx, \qquad t \notin \R.
\end{equation} 
With $\mathcal A_N[t]$ the differential operator specified in (\ref{DE of real density}), but with respect to $t$ rather than $x$, we have that $W(t)$ satisfies the inhomogeneous differential equation
\begin{equation}\label{1.33}
\mathcal A_N[t] W(t) = (1 + 4 N - 2 t^2) M_{0,N}^{\rm r} - 6 M_{N,2}^{\rm r}.
\end{equation} 
\end{cor}

\medskip

\begin{rem}
In \cite[Cor. 1.5]{By23b}, it was found that the moment generating function $u(t)$ satisfies a seventh-order differential equation.  
Indeed, it can be observed that the differential equation in \cite[Cor. 1.5]{By23b} can be further simplified to 
\begin{align}
\begin{split}
0&= (t^3\,\partial_t^3-6t^2\,\partial_t^2+15t\,\partial_t-15) \circ \,D_N[t] u (t)
\\
&=  (t \,\partial_t-a)\circ(t \,\partial_t-b)\circ(t \,\partial_t-c) \circ D_N[t] \, u (t) , 
\end{split}
\end{align}
where $( a,b,c )$ is a permutation of $(1,3,5)$.
\end{rem}

\medskip

The proofs of the above results are given in Section \ref{S2}. In Section \ref{S3} we link up the large $N$ asymptotic expansion of the moments Theorem \ref{Thm_large N} with large $N$ asymptotic expansions that can be deduced for the Fourier-Laplace transform $u(t)$ and the Stieltjes transform $W(t)$. 

\section{Proofs}\label{S2}

\subsection{Proof of Theorem \ref{Thm_large N}}
To begin we
recall a particular asymptotic formula of ${}_2F_1$, telling us that as $\lambda \to \infty$ ($\lambda \in \R$), 
\begin{equation}\label{2F1 asymp}
\pFq{2}{1}{a,b}{c+\lambda}{z}	= \frac{ \Gamma(c+\lambda)  }{ \Gamma(c-b+\lambda) } \sum_{s=0}^{m-1} q_s(z) \frac{\Gamma(b+s)}{\Gamma(b)} \,\lambda^{-s-b}+{\rm O}(\lambda^{-m-b});
	\end{equation}	
see \cite[Eq.(15.12.3)]{NIST}.	 Here $q_0(z)=1$ and $q_s(z)$ when $s=1,2,\ldots$ are defined by the generating function
	\begin{equation}
		\Big( \frac{e^t-1}{t} \Big)^{b-1} e^{t (1-c) } (1-z+ze^{-t})^{-a} =\sum_{s=0}^\infty q_s(z) \, t^s.
	\end{equation}
Using this with $a=1,b=-1/2, c=0$ and $\lambda=N$ in (\ref{1.12}) gives (\ref{M0N asymp}).

In preparation for deducing the large $N$ asymptotic form of $M_{2,N}^{\rm r}$ an evaluation formula analogous to (\ref{1.12}) is required.

\begin{prop} \label{Prop_moment p2}
We have 
\begin{align}\label{AM2}
\begin{split}
M_{2,N}^{ \rm r } 
&=  \sqrt{ \frac{2}{\pi} } \frac{\Gamma(N+\frac32)}{(N-1)!} \bigg( \frac{1}{2N}  \pFq{2}{1}{2,-\frac12}{N+1}{\frac12} + \frac{1}{3}  \pFq{2}{1}{1,-\frac32}{N}{\frac12} \bigg) + \frac{N}{2}. 
\end{split}
\end{align}
\end{prop}

\begin{proof}
We will assume temporarily the validity of the evaluation formula given by the $p=1$ case of $M_{2p,N}^{\rm r}$ in Theorem~\ref{Thm_GinOE moments}
(this will be proved for all non-negative $p$ in the next section). This gives 
\begin{equation}
\begin{split}
M_{2,N}^{ \rm r }   =  \frac{1}{(N-2)!}  \sum_{k=0}^\infty \frac{1}{ 2^{k+\frac12}  }  \frac{  \Gamma(k+\frac32)  \Gamma(N-k+\frac12)  }{  \Gamma(k+\frac12) \, \Gamma(-k+\frac52) }   +  N\, \mathbbm{1}_{ \{ N: \textup{ odd} \} }. 
\end{split}
\end{equation}

We note that 
\begin{align*}
&\quad  \sum_{k=0}^\infty \frac{1}{ 2^{k+\frac12}  }  \frac{  \Gamma(k+\frac32)  \Gamma(N-k+\frac12)  }{  \Gamma(k+\frac12) \, \Gamma(-k+\frac52) }  =  \sum_{k=0}^\infty \frac{1}{ 2^{k+\frac12}  }  \frac{ (k+\frac12)  \Gamma(N-k+\frac12)  }{  \Gamma(-k+\frac52) }  
\\
&=  \sum_{k=0}^\infty \frac{1}{ 2^{k+\frac12}  }  \frac{ (k+\frac12)  \Gamma(N-k+\frac12)  }{  \Gamma(-k+\frac52) }  =  \sum_{k=0}^\infty \frac{1}{ 2^{k+\frac12}  }  \frac{ ( (k+\frac32)-1)  \Gamma(N-k+\frac12)  }{  \Gamma(-k+\frac52) }  
\\
&=   \frac{\Gamma(N+\tfrac12)}{ 6\sqrt{\pi}  } \pFq{2}{1}{ -\frac32, -N-\frac12 }{ -N+\frac12 }{-1}   +\frac{ \Gamma(N-\tfrac12) }{ 2\sqrt{\pi} }  \pFq{2}{1}{ -\frac12, -N-\frac12 }{ -N+\frac32 }{-1}
\\
&=  \frac{ \sqrt{\pi} (-1)^N }{ 6 \, \Gamma(-N+\frac12)   }   \pFq{2}{1}{ -\frac32, -N-\frac12 }{ -N+\frac12 }{-1}  - \frac{ \sqrt{\pi}(-1)^N }{ 2\, \Gamma(-N+\tfrac32) }  \pFq{2}{1}{ -\frac12, -N-\frac12 }{ -N+\frac32 }{-1}.
\end{align*}
Recall the linear transformation \cite[Eq.(15.8.5)]{NIST}
\begin{align}
\begin{split} \label{linear transform 2F1}
\frac{ \sin(\pi ( c-a-b ) )  }{ \pi \,\Gamma(c) } \pFq{2}{1}{a,b}{c}{z} &= \frac{ z^{-a} }{ \Gamma(c-a)\Gamma(c-b) \Gamma(a+b-c+1) } \pFq{2}{1}{a,a-c+1}{a+b-c+1}{1- \frac{1}{z} }
\\
&\quad - \frac{(1-z)^{ c-a-b }z^{a-c} }{ \Gamma(a)\Gamma(b) \Gamma( c-a-b+1) } \pFq{2}{1}{c-a,1-a}{ c-a-b+1}{1-\frac{1}{z}}.
\end{split}
\end{align}
Here, let us mention that the regularised hypergeometric function $ {}_2 \textbf{F}_1 $ in  \cite[Eq.(15.8.5)]{NIST} is given by 
 $$
   {}_2 \textbf{F}_1 (a,b;c;z)= \frac{1}{\Gamma(c)} {}_2 F_1(a,b;c;z). 
   $$
Using \eqref{linear transform 2F1} with $a=-N-1/2, b=-3/2,c=-N+1/2$, we have 
\begin{align*}
&\quad \frac{ 1 }{  \pi  \Gamma(-N+\frac12)   }     \pFq{2}{1}{  -N-\frac12 , -\frac32 }{ -N+\frac12 }{-1} =   \frac{ 2^{ 5/2 } }{ \Gamma(-N-\frac12)\Gamma(-\frac32) \Gamma( \frac{7}{2}) }  \pFq{2}{1}{1,N+\frac32}{ \frac72 }{2} 
\\
&=  -\frac{(-1)^{N} 2^{ 5/2 } \Gamma(N+\frac32)  }{ \pi \Gamma(-\frac32) \Gamma( \frac{7}{2}) }  \pFq{2}{1}{1,N+\frac32}{ \frac72 }{2}
= -\frac{(-1)^{N} 2^{ 7/2 } \Gamma(N+\frac32)  }{ 5 \pi^2 }  \pFq{2}{1}{1,N+\frac32}{ \frac72 }{2},
\end{align*}
which leads to 
$$
\frac{ \sqrt{\pi} (-1)^N }{ 6 \, \Gamma(-N+\frac12)   }     \pFq{2}{1}{  -N-\frac12 , -\frac32 }{ -N+\frac12 }{-1}  =  - \frac{4}{15}\sqrt{ \frac{2}{\pi}  } \Gamma(N+\tfrac32)  \pFq{2}{1}{1,N+\frac32}{ \frac72 }{2}. 
$$
Here, we have used $1/\Gamma(-N+2)=0$, the reflection formula of Gamma function 
\begin{equation}\label{Gamma reflection}
\Gamma(z)\Gamma(1-z)=\pi/\sin(\pi z)  
\end{equation}
and 
$$
\Gamma(-\tfrac32)= \frac{4\sqrt{\pi}}{3}, \qquad \Gamma( \tfrac72 )= \frac{15 \sqrt{\pi}}{8} 
$$
in each identity. 
Similarly, we have 
\begin{align*}
 \frac{ \sqrt{\pi}(-1)^N }{ 2\, \Gamma(-N+\tfrac32) }  \pFq{2}{1}{ -N-\frac12,-\frac12 }{ -N+\frac32 }{-1} &=  - \frac{8}{15}\sqrt{ \frac{2}{\pi}  } \Gamma(N+\tfrac32)  \pFq{2}{1}{2,N+\frac32}{ \frac72 }{2}.
\end{align*}
Furthermore, again using \eqref{linear transform 2F1} after removing the removable singularities, it follows that 
\begin{align*}
&\quad - \frac{4}{15}\sqrt{ \frac{2}{\pi}  } \Gamma(N+\tfrac32)  \pFq{2}{1}{1,N+\frac32}{ \frac72 }{2} =  -\frac{1}{ \sqrt{2} }  \frac{ \Gamma(N+\tfrac32)  }{ \Gamma(\frac{7}{2}) } \pFq{2}{1}{1,N+\frac32}{ \frac72 }{2}
\\
&= \sqrt{ \frac{2}{\pi} } \frac{ \Gamma(N+\frac32) }{3 (N-1)  }  \pFq{2}{1}{1,-\frac32}{N}{\frac12} + \frac{(-1)^N (N-2)! }{ 8 }\pFq{2}{1}{ \frac52,0 }{ 2-N }{\frac12} 
\\
&= \sqrt{ \frac{2}{\pi} } \frac{ \Gamma(N+\frac32) }{3 (N-1)  }  \pFq{2}{1}{1,-\frac32}{N}{\frac12} + \frac{(-1)^N (N-2)! }{ 8 } 
\end{align*}
and 
\begin{align*}
&\quad   \frac{8}{15}\sqrt{ \frac{2}{\pi}  } \Gamma(N+\tfrac32)  \pFq{2}{1}{2,N+\frac32}{ \frac72 }{2} =  \sqrt{2}  \frac{ \Gamma(N+\tfrac32)  }{ \Gamma(\frac{7}{2}) } \pFq{2}{1}{2,N+\frac32}{ \frac72 }{2} 
\\
&= \frac{1}{ \sqrt{2\pi} } \frac{ \Gamma(N+\frac32)  }{N(N-1)} \pFq{2}{1}{2,-\frac12}{N+1}{\frac12}  +\frac{ (-1)^N (N-1)!}{2} \pFq{2}{1}{\frac32,-1}{1-N}{\frac12}
\\
&= \frac{1}{ \sqrt{2\pi} } \frac{ \Gamma(N+\frac32)  }{N(N-1)} \pFq{2}{1}{2,-\frac12}{N+1}{\frac12}  +\frac{ (-1)^N (4N-1)\, (N-2)!}{8}. 
\end{align*}
Therefore we obtain
\begin{align*}
&\quad \sqrt{ \frac{2}{\pi}  } \frac{4}{15}  \frac{\Gamma(N+\tfrac32) }{(N-2)!} \bigg( 2\,\pFq{2}{1}{2,N+\frac32}{ \frac72 }{2}-   \pFq{2}{1}{1,N+\frac32}{ \frac72 }{2} \bigg)
\\
&=  \frac{1}{ \sqrt{2\pi} } \frac{ \Gamma(N+\frac32)  }{N!} \pFq{2}{1}{2,-\frac12}{N+1}{\frac12} +  \sqrt{ \frac{2}{\pi} } \frac{ \Gamma(N+\frac32) }{3 (N-1)!  }  \pFq{2}{1}{1,-\frac32}{N}{\frac12} + (-1)^N\frac{N}{2}
\\
&= \sqrt{ \frac{2}{\pi} } \frac{\Gamma(N+\frac32)}{(N-1)!} \bigg( \frac{1}{2N}  \pFq{2}{1}{2,-\frac12}{N+1}{\frac12} + \frac{1}{3}  \pFq{2}{1}{1,-\frac32}{N}{\frac12} \bigg) + (-1)^N\frac{N}{2}, 
\end{align*}
which completes the proof.
\end{proof}

Focusing now on (\ref{AM2}) we note from \eqref{2F1 asymp} with $a=2,b=-1/2, c=1$ that
 \begin{align*}
 \frac{1}{ \sqrt{2\pi} } \frac{ \Gamma(N+\frac32)  }{N!} \pFq{2}{1}{2,-\frac12}{N+1}{\frac12} & 
 \sim N\sqrt{ \frac{2}{\pi} N }  \sum_{s=1}^{m-1} \frac{\alpha_{s-1}}{2} \frac{ \Gamma(-3/2+s) }{ \Gamma(-1/2) } N^{-s}, 
 \end{align*}
 where 
\begin{equation}
		4\Big( \frac{e^t-1}{t} \Big)^{-3/2} (1+e^{-t})^{-2} =\sum_{s=0}^\infty \alpha_s\, t^s.
	\end{equation}
Also, using \eqref{2F1 asymp} with $a=1,b=-3/2,c=0$, we have
\begin{align*}
 \sqrt{ \frac{2}{\pi} } \frac{ \Gamma(N+\frac32) }{3 (N-1)!  }  \pFq{2}{1}{1,-\frac32}{N}{\frac12}  \sim N \sqrt{ \frac{2}{\pi} N }  \sum_{s=0}^{m-1} \frac{\beta_s}{3} \frac{ \Gamma(-3/2+s) }{ \Gamma(-3/2) } N^{-s}
\end{align*}
where 
\begin{equation}
2\Big( \frac{e^t-1}{t} \Big)^{-5/2} e^{t } (1+e^{-t})^{-1} =\sum_{s=0}^\infty \beta_s \, t^s.
\end{equation}
Therefore we have shown that
\begin{align*}
&\quad \sqrt{ \frac{2}{\pi} } \frac{\Gamma(N+\frac32)}{(N-1)!} \bigg( \frac{1}{2N}  \pFq{2}{1}{2,-\frac12}{N+1}{\frac12} + \frac{1}{3}  \pFq{2}{1}{1,-\frac32}{N}{\frac12} \bigg) 
\\
&\ \sim  N \sqrt{ \frac{2}{\pi} N } \bigg( \frac{\beta_0}{3}+ \sum_{s=1}^{m-1} \frac{ \beta_s-\alpha_{s-1} }{ 3 }  \frac{ \Gamma(-3/2+s) }{ \Gamma(-3/2) } N^{-s} \bigg). 
\end{align*}
Substituting in (\ref{AM2}) the expansion \eqref{M2N asymp} follows.
\hfill $\square$

\subsection{Proof of Theorem~\ref{Thm_GinOE moments}}

In relation to Theorem~\ref{Thm_GinOE moments},
note that the spectral moments \eqref{M 2pN c}  of complex eigenvalues follow from \eqref{spectral moments for real and complex}, \eqref{real+cplx} and their real counterparts \eqref{M 2pN r}. 
Thus it suffices to show \eqref{M 2pN r}. The more general setting of Proposition \ref{P1.3} we will be assumed, requiring only that ${\rm Re} \, p > -1/2$.

By using \cite[Cor.~4.1]{EKS94}, we have  
\begin{equation}
\sum_{ N=1 }^\infty \rho_N^{ \rm r  }(x) z^N = \mathcal{F}(z,x),
\end{equation}
where 
\begin{equation}
\mathcal{F}(z,x):= \frac{z }{ \sqrt{2\pi} } \bigg( e^{ -\frac{x^2}{2} } +\frac{z}{1-z} e^{ (z-1) x^2 } \bigg)  + \frac{ z^2  |x| }{ 2 } e^{ \frac{(z^2-1)x^2}{2} } \bigg( \erf\Big( \frac{ z |x| }{ \sqrt{2} } \Big) + \erf\Big( \frac{ (1-z) |x| }{ \sqrt{2} } \Big)  \bigg). 
\end{equation}
It then follows that 
\begin{equation}
\sum_{ N=0 }^\infty M_{2p,N}^{ \rm r } \, z^N =  \int_\R |x|^{2p} \mathcal{F}(z,x)\,dx. 
\end{equation}

Since 
\begin{equation*}
\int_\R |x|^{2p} e^{-\frac{x^2}{2} }\,dx  =  2^{ p+\frac12 } \, \Gamma\Big(p+\frac12 \Big), \qquad  \int_\R |x|^{2p} e^{ (z-1)x^2 }\,dx  =  (1-z)^{ -p-\frac12 } \Gamma\Big(p+\frac12 \Big), 
\end{equation*}
where in the second integral it is assumed $|z|<1$, we have 
\begin{align}
\frac{z }{ \sqrt{2\pi} } \int_\R   |x|^{2p} \bigg( e^{ -\frac{x^2}{2} } +\frac{z}{1-z} e^{ (z-1) x^2 } \bigg)\,dx &= \frac{z }{ \sqrt{2\pi} } \bigg(  2^{ p+\frac12 } + \frac{z}{(1-z)^{ p+\frac32 }}   \bigg) \Gamma\Big(p+\frac12 \Big).  
\end{align}

On the other hand, by using the expansion \cite[Eq.(7.6.2)]{NIST} 
\begin{equation}
\erf(z) = \frac{2}{ \sqrt{\pi} } e^{-z^2} \sum_{k=0}^\infty \frac{ 2^k z^{2k+1} }{ (2k+1)!! },
\end{equation}
we have
\begin{align}
\begin{split}
& \int_\R |x|^{2p} \frac{ z^2  |x| }{ 2 } e^{ \frac{(z^2-1)x^2}{2} } \bigg( \erf\Big( \frac{ z |x| }{ \sqrt{2} } \Big) + \erf\Big( \frac{ (1-z) |x| }{ \sqrt{2} } \Big)  \bigg) \,dx 
\\
&\qquad =  \frac{z^2}{2} \sqrt{ \frac{2}{\pi} }  \sum_{k=0}^\infty \int_\R  \frac{ |x|^{2p+2k+2} }{ (2k+1)!! } \bigg( z^{2k+1} e^{ -\frac{x^2}{2} } +(1-z)^{2k+1} e^{ (z-1)x^2 } \bigg) \,dx   
\\
&\qquad = \frac{z^2}{2} \sqrt{ \frac{2}{\pi} }  \sum_{k=0}^\infty\frac{1}{(2k+1)!!} \bigg( z^{2k+1} 2^{ p+k+\frac32 } + (1-z)^{ k -p-\frac12 } \bigg)  \Gamma\Big( p+k+\frac32 \Big).  
\end{split}
\end{align}
Combining the above, we have 
\begin{align*}
\int_\R |x|^{2p} \mathcal{F}(z,x)\,dx & =  \frac{z }{ \sqrt{2\pi} } \bigg(  2^{ p+\frac12 } + \frac{z}{(1-z)^{ p+\frac32 }}   \bigg) \Gamma\Big(p+\frac12 \Big)  
\\
&\quad + \frac{z^2}{ \sqrt{2\pi} } \sum_{k=0}^\infty\frac{1}{(2k+1)!!} \bigg( z^{2k+1} 2^{ p+k+\frac32 } + (1-z)^{ k -p-\frac12 } \bigg)  \Gamma\Big( p+k+\frac32 \Big) ,
\end{align*}
which can be written as 
\begin{equation}
\int_\R |x|^{2p} \mathcal{F}(z,x)\,dx= \frac{z^2}{ \sqrt{2\pi} } \sum_{k=0}^\infty\frac{1}{(2k-1)!!} \bigg( z^{2k-1} 2^{ p+k+\frac12 } + (1-z)^{ k -p-\frac32 } \bigg)  \Gamma\Big( p+k+\frac12 \Big)  .
\end{equation}

By using 
\begin{align*}
(1-z)^{k-p-\frac32} = \sum_{l=0}^\infty \binom{k-p-3/2}{l} (-z)^{l} = \sum_{l=0}^\infty \frac{ \Gamma(k-p-\frac12) }{ l! \, \Gamma(k-p-l-\frac12) }  (-z)^{l}, 
\end{align*}
we obtain
\begin{align}
\begin{split}
\int_\R |x|^{2p} \mathcal{F}(z,x)\,dx &= \frac{1}{\sqrt{2\pi}} \sum_{k=0}^\infty \frac{ \Gamma(k+p+\frac12) 2^{p+k+\frac12} }{ (2k-1)!! } z^{2k+1} 
\\
&\quad + \frac{1}{\sqrt{2\pi}} \sum_{N=0}^\infty \bigg( \sum_{k=0}^\infty \frac{ (-1)^{N} \Gamma(k+p+\frac12) \Gamma(k-p-\frac12) }{(2k-1)!!\,(N-2)!\, \Gamma(k-p-N+\frac32)} \bigg) z^{N}.
\end{split}
\end{align}
This gives rise to 
\begin{equation}
\begin{split}
M_{2p,N}^{ \rm r }  & = \frac{1}{\sqrt{2\pi}}  \frac{(-1)^N}{(N-2)!}  \sum_{k=0}^\infty \frac{  \Gamma(k+p+\frac12) \Gamma(k-p-\frac12) }{(2k-1)!! \, \Gamma(k-p-N+\frac32)} +\frac{1}{\sqrt{2\pi}}   \frac{ \Gamma(p+N/2) 2^{p+N/2} }{ (N-2)!! }  \mathbbm{1}_{ \{ N: \textup{ odd} \} }
\\
&  =  \frac{1}{(N-2)!}  \sum_{k=0}^\infty \frac{1}{ 2^{k+\frac12}  }  \frac{  \Gamma(k+p+\frac12)  \Gamma(N+p-k-\frac12)  }{  \Gamma(k+\frac12) \, \Gamma(p-k+\frac32) }   +  2^p \frac{ \Gamma(p+N/2)  }{ \Gamma(N/2) } \mathbbm{1}_{ \{ N: \textup{ odd} \} }. 
\end{split}
\end{equation}
Note here that 
\begin{align*}
\frac{1}{\sqrt{2\pi}}   \frac{ \Gamma(p+N/2) 2^{p+N/2} }{ (N-2)!! }  &=  \frac{ \Gamma(p+N/2) 2^{p} }{ \Gamma(N/2) }   
\end{align*}
and that 
\begin{align*}
\frac{1}{\sqrt{2\pi}}   \frac{  \Gamma(k+p+\frac12) \Gamma(k-p-\frac12) }{(2k-1)!! \, \Gamma(k-p-N+\frac32)} & =  \frac{1}{ 2^{k+\frac12}  }  \frac{  \Gamma(k+p+\frac12) \Gamma(k-p-\frac12) }{  \Gamma(k+\frac12) \, \Gamma(k-p-N+\frac32)} 
\\
& = \frac{ (-1)^{N} }{ 2^{k+\frac12} }   \frac{  \Gamma(k+p+\frac12)  \Gamma(N+p-k-\frac12)  }{  \Gamma(k+\frac12) \, \Gamma(p-k+\frac32) },
\end{align*}
where to obtain the final line use has been made of \eqref{Gamma reflection}. 
Now the expression \eqref{M 2pN r} follows from the definition \eqref{def of gen hypergeometric} of the generalised hypergeometric function. 
\hfill $\square$

\subsection{Proof of Proposition \ref{P1.5}}

We follow a strategy introduced in 
\cite[Proofs of Thms.~4 and 9]{WF14} in relation to the differential equations satisfied by the eigenvalue density for the GUE and GOE.
Due to the symmetry $x \mapsto -x$, it suffices to consider the case $x>0.$
Let 
\begin{align*}
f(x) = \frac{d}{dx} \bigg[  \Gamma(N-1,x^2) +  2^{(N-3)/2} e^{ -\frac{x^2 }{2} } x^{N-1} \gamma\Big( \frac{N-1}{2}, \frac{x^2}{2} \Big)  \bigg],
\end{align*}
so that according to \eqref{real density}, the function $f(x)$ is proportional to $\frac{d}{dx}\rho_N(x)$. In terms of $f(x)$, the differential equation of the proposition reads
\begin{equation}\label{2.18}
 x^2 f''(x) + x ( 3 x^2-3N+4) f'(x) + ( 2 x^2-2N+1 ) (  x^2-N+2 ) f(x)=0. 
\end{equation}

Using 
\begin{align*}
\frac{d}{dx}   \Gamma(N-1,x^2) = -2\, x^{2N-3} e^{-x^2} ,\qquad   2^{(N-3)/2} e^{ -\frac{x^2 }{2} } \frac{d}{dx}  \gamma\Big( \frac{N-1}{2}, \frac{x^2}{2} \Big) = x^{N-2} e^{-x^2}, 
\end{align*}
we have
\begin{align} \label{nk1}
\begin{split}
f(x)  &= - x^{2N-3} e^{-x^2} +  2^{(N-3)/2} e^{ -\frac{x^2 }{2} } x^{N-2} \Big( -x^2+ N-1 \Big)   \gamma\Big( \frac{N-1}{2}, \frac{x^2}{2} \Big) 
\\
& =: - a(x) + (-x^2 + N - 1) b(x).
\end{split}
\end{align}
Similarly, it follows that 
\begin{align}\label{nk2} 
\begin{split}
f'(x) &= x^{2N-3} \Big( x -\frac{ N-2 }{ x } \Big) e^{-x^2}
 \\
&\quad +  2^{(N-3)/2} e^{ -\frac{x^2 }{2} } x^{N-2} \Big(  x^{3}-(2N-1)x+\frac{(N-1)(N-2)}{x}  \Big)   \gamma\Big( \frac{N-1}{2}, \frac{x^2}{2} \Big) 
\\
& = \Big( x -\frac{ N-2 }{ x } \Big) a(x)  + \Big(  x^{3}-(2N-1)x+\frac{(N-1)(N-2)}{x}  \Big) b(x),
\end{split}
\end{align}
and 
\begin{align}\label{nk3}
\begin{split}
f''(x) &=  x^{2N-3} \Big( -x^{2} +(2N-5) -\frac{(N-2)(N-3)}{x^2}  \Big) e^{-x^2} 
\\
&\quad +  2^{(N-3)/2} e^{ -\frac{x^2 }{2} }  x^{N-2} \Big(  -x^{4}+3N\,x^{2}-3(N-1)^2  +\frac{(N-1) (N-2)(N-3)}{x^2}   \Big)   \gamma\Big( \frac{N-1}{2}, \frac{x^2}{2} \Big)  
\\
& = \Big( -x^{2} +(2N-5) -\frac{(N-2)(N-3)}{x^2}  \Big) a(x) 
\\
&\quad + \Big(  -x^{4}+3N\,x^{2}-3(N-1)^2  +\frac{(N-1) (N-2)(N-3)}{x^2}   \Big) b(x).
\end{split}
\end{align}
Combining all of the above, the desired differential equation (\ref{2.18}). The intermediate working is best carried out using computer algebra for efficiency and accuracy.
\hfill $\square$

\medskip 

We mention that at the beginning, the way to derive the exact form of the differential equation \eqref{2.18} is to use (\ref{nk1}) and (\ref{nk2}) regarded as a linear system to solve for $a(x)$, $b(x)$ in terms of $f(x)$ and $f'(x)$.
With this done, substituting in (\ref{nk3}) gives (\ref{2.18}).

\subsection{Proof of Corollary \ref{C1.6}}

To deduce the three term recurrence (\ref{1.17}) with $p$ in general complex, we multiply the differential equation by $x^p$ and integrate over $x \in (0,\infty)$.
Carrying out the integration using integration by parts gives (\ref{1.17}), provided $\re (p)$ is large enough so that all terms requiring evaluation at $x=0$ vanish. Analytic continuation removes the need for such a restriction.

In relation to the three term recurrence for the ${}_3F_2$ function (\ref{1.28}), we
 notice by direct computations that
 \begin{align*}
2(2p+5) 2^{p+2}\frac{ \Gamma(p+2+N/2) }{ \Gamma(N/2) } & = (2p+3) (6p+4N+7) 2^{p+1} \frac{ \Gamma(p+1+N/2) }{ \Gamma(N/2) } 
\\
&\quad -  (2p+1) (2p+N)(2p+2N+1) 2^{p}  \frac{ \Gamma(p+N/2) }{ \Gamma(N/2)  }. 
 \end{align*}
Therefore by \eqref{M 2pN r}, one can observe that the recursion formula (\ref{1.17}) implies (\ref{1.28}).
\hfill $\square$

\subsection{Proof of Corollary \ref{Cor_DE of MGF}}

Note that 
\begin{align*}
\frac{d^k}{dt^k}u(t) = \int_\R e^{tx} x^k \rho_N^{ \rm r }(x)\,dx. 
\end{align*}
Then the differential equation \eqref{DE MGF}  follows from Proposition~\ref{P1.3}, after multiplying by $e^{tx}$ and integrating over $x \in \R$ using integrating by parts. For the latter, we note
\begin{align*}
& \frac{d}{dx} \Big[   ( 2 x^2-2N+1 ) (  x^2-N+2 ) e^{tx}  \Big] -\frac{d^2}{dx^2} \Big[ x(3x^2-3N+4) e^{tx} \Big] + \frac{d^3}{dx^3}\Big[ x^2e^{tx} \Big]
\\
&\qquad= \Big ( 2 t x^4 - ( 3 t^2-8 ) x^3 + t(t^2 - 4 N   -13 ) x^2 +  (  (3 N+2) t^2 -8N-8  ) x + (2N^2+N)t \Big ) e^{tx}.   
\end{align*}

In relation to the Stieltjes transform, we first separate out the coefficients  of $\partial_x^2$ and $\partial_x$ in $\mathcal A_N[x]$ so that the differential equation (\ref{DE of real density}) reads
$$
\Big ( x^2 \partial_x^3 + R(x) \partial_x^2 + S(x) \partial_x \big ) \rho_N^{\rm r}(x) = 0,
$$
with
$$
R(x) = x (3 x^2 - 3N + 4), \qquad
S(x) = (2 x^2 - 2N + 1) (x^2 - N + 2).
$$
Next we manipulate this equation so that it takes the form
$$
\Big ( t^2 \partial_x^3 + R(t) \partial_x^2 + S(t) \partial_x \Big ) \rho_N^{\rm r}(x) 
= \Big ( (t^2 - x^2) \partial_x^3 + (R(t) - R(x)) \partial_x^2 +
(S(t) - S(x)) \partial_x \Big ) \rho_N^{\rm r}(x).
$$
Multiplying through by $1/(t-x)$ and integrating over $x \in \R$, the LHS is readily identified as $\mathcal A_N[t]$ after integration by parts. On the RHS the term $1/(t-x)$ can be cancelled with factors in the coefficients, which then reduce to polynomials symmetric in $t$ and $x$. Integration by parts requires that these polynomials be differentiated with respect to $x$ a suitable number of times, with the result giving the RHS of (\ref{1.33}).
\hfill $\square$

\section{Links between large $N$ expansions}\label{S3}

\subsection{Asymptotic expansion of the moment generating function}

Let us define the rescaled moment generating function
\begin{equation}
\widetilde{u}(t):= \frac{1}{\sqrt{N}} u \Big(\frac{t}{\sqrt{N}}\Big) = \int_\R e^{tx} \rho_N^{ \rm r }( \sqrt{N} x) \,dx.  
\end{equation}
Then \eqref{DE MGF} gives rise to 
\begin{equation}\label{3.2}
\bigg( \mathcal{D}_0[t] + \frac{\mathcal{D}_1[t]}{N}+ \frac{\mathcal{D}_2[t]}{N^2} \bigg) \widetilde{u}(t)=0,  
\end{equation}
where 
\begin{align}
\begin{split} \label{mathcal D0}
\mathcal{D}_0[t]& :=  2 t \, \partial_t^4 + 8 \, \partial_t^3  -4t \, \partial_t^2 -8 \, \partial_t + 2t  
\end{split}
\\
\mathcal{D}_1[t]& := -3t^2 \, \partial_t^3 -13 t \, \partial_t^2 +( 3t^2-8  ) \, \partial_t +t,  
\label{mathcal D1} 
\\
\mathcal{D}_2[t]  & :=  t^3 \, \partial_t^2+2t^2 \, \partial_t . 
\end{align}

As is consistent with the expansion of the even integer moments Theorem \ref{Thm_large N}, introduce the large $N$ expansion
\begin{equation}\label{3.6}
\widetilde{u}(t)= \sum_{k=0}^\infty \bigg ( \frac{ \widetilde{u}_{ (k) }(t) }{ N^k } + \frac{ \widetilde{u}_{ (k+1/2) }(t) }{ N^{k+1/2} } \bigg ).
\end{equation}
Use the expansion coefficients therein to introduce a sequence of smoothed densities $\{ r_{(k)}(x), r_{(k+1/2)}(x) \}$ by the requirement that
\begin{equation}\label{3.8}
\widetilde{u}_{ (k) }(t) =
\int_{\R} e^{tx} r_{(k)}(x) \, dx, \qquad
\widetilde{u}_{ (k+1/2) }(t) =
\int_{\R} e^{tx} r_{(k+1/2)}(x) \, dx.
\end{equation}
Note that then, in a formal sense, and with the LHS interpreted as being always begin integrated against a smooth function, we then have
\begin{equation}
\rho_N^{ \rm r }(\sqrt{N}x) = \sum_{k=0}^\infty \bigg ( \frac{ {r}_{ (k) }(x) }{ N^k } + \frac{ {r}_{ (k+1/2) }(x) }{ N^{k+1/2} } \bigg ).
\end{equation}
We know from \cite[displayed equation below Eq.(3.5) and Eq.(3.8)]{BF23} that
\begin{equation} \label{r0 r12}
r_{(0)}(x) = \frac{1}{\sqrt{2\pi}} \, \mathbbm{1}_{(-1,1)}(x), \qquad r_{(1/2)}(x) = \frac14\Big( \delta(x-1)+\delta(x+1) \Big) , 
\end{equation}  
and so
\begin{equation} \label{u0 1/2 sinh cosh}
\widetilde{u}_{(0)}(t) :=  \sqrt{ \frac{2}{\pi} } \frac{ \sinh(t) }{t}, \qquad \widetilde{u}_{(1/2)}(t) :=   \frac{\cosh(t)}{2}. 
\end{equation}
Note also that by \eqref{M0N asymp} and \eqref{M2N asymp}, for $k \ge 1$,
\begin{equation} \label{u0 origin}
\widetilde{u}_{(k)}(0)=  \sqrt{ \frac{2}{\pi}  } \, a_k, \qquad  \widetilde{u}''_{(k)}(0)=  \sqrt{ \frac{2}{\pi}  } \, b_k
\end{equation}
and 
\begin{equation}
\widetilde{u}_{(k+1/2)}(0)= \widetilde{u}''_{(k+1/2)}(0)= 0. 
\end{equation}

Scaling (\ref{3.2}), $t \mapsto t/\sqrt{N}$, and substituting  (\ref{3.6}) shows
\begin{equation} \label{recursion u tilde k}
\mathcal{D}_0[t]\, \widetilde{u}_{ (k) }(t) +  \mathcal{D}_1[t]\, \widetilde{u}_{ (k-1) }(t) +  \mathcal{D}_2[t]\, \widetilde{u}_{ (k-2) }(t) =0 
\end{equation}
and 
\begin{equation} \label{recursion u tilde k 12}
\mathcal{D}_0[t]\, \widetilde{u}_{ (k+1/2) }(t) +  \mathcal{D}_1[t]\, \widetilde{u}_{ (k-1/2) }(t) +  \mathcal{D}_2[t]\, \widetilde{u}_{ (k-3/2) }(t) =0, 
\end{equation}
with the convention that $\widetilde{u}_j \equiv 0$ if $j <0.$
One observes that the differential operator $\mathcal{D}_0[t]$ has the factorisations
\begin{equation}
\mathcal{D}_0[t]  =2 ( t \, \partial_t^2+4\, \partial_t -t ) \circ(\pa_t^2-1)=2 ( \partial_t^2-1 )  \circ( t\,\partial_t^2+2\, \partial_t-t ).  
\end{equation}
From the explicit functional forms of $\widetilde{u}_{(0)}(t)$  and $\widetilde{u}_{(1/2)}(t)$ (\ref{u0 1/2 sinh cosh}) it is observed that both are annihilated by $\mathcal{D}_0[t]$. 
Indeed, the general even solution to $\mathcal{D}_0[t] f(t)=0$ is of the form
$$
f(t) = c_0 \cosh(t) + c_1 \frac{\sinh(t)}{t}, \qquad c_0,c_1 \in \R.  
$$

By taking $k=1$ in \eqref{recursion u tilde k}, 
\begin{align*}
\mathcal{D}_0[t] \,\widetilde{u}_1(t)+ \mathcal{D}_1[t] \,\widetilde{u}_0(t) = \mathcal{D}_0[t] \,\widetilde{u}_1(t) - 6  \sqrt{ \frac{2}{\pi} } \sinh(t) =0.
\end{align*}
By solving this differential equation with the initial condition \eqref{u0 origin}, we have
\begin{align}
\widetilde{u}_{(1)}(t) =  \sqrt{ \frac{2}{\pi} } \frac{3}{8}  \bigg( t \sinh(t) - \cosh(t)  \bigg) . 
\end{align}
Similarly, it follows that 
\begin{align}
\widetilde{u}_{(2)}(t) &= \sqrt{ \frac{2}{\pi} }  \frac{1}{384} \bigg(  (23t^2+9)\,t\, \sinh(t)-(26t^2+9) \cosh(t)  \bigg),
\\
\widetilde{u}_{(3)}(t) &= \sqrt{ \frac{2}{\pi} } \frac{1}{15360} \bigg(  ( 91 t^4- 285 t^2 -405) \,t\,\sinh(t) -5 ( t^4- 84 t^2 -81) \cosh(t)  \bigg).
\end{align}
In general, one can observe that $\widetilde{u}_k$ is of the form
\begin{equation}
\widetilde{u}_{(k)}(t)=  \sqrt{ \frac{2}{\pi} } \bigg( P_{k,1}(t)\,t\, \sinh(t) + P_{k,2}(t) \,\cosh(t) \bigg),
\end{equation}
where $P_{k,1}$ and $P_{k,2}$ are some even polynomials of degree $k+1.$ The corresponding quantities in the expansion (\ref{3.6}) are then
\begin{equation}\label{3.20}
{r}_{(k)}(x)=  \frac{1}{ \sqrt{2 \pi} }  \bigg (  P_{k,1}(-\partial_x)(-\partial_x) \Big ( \delta(x-1) - \delta(x+1) \Big ) +
P_{k,2}(-\partial_x) \Big ( \delta(x-1) + \delta(x+1) \Big )
\bigg )
\end{equation}
as can be checked from (\ref{3.8}).

Regarding the half integer coefficients in (\ref{3.6}), we also have 
\begin{align}
\widetilde{u}_{(3/2)}(t) & = \frac18 \bigg( t^2\, \cosh(t) -t \, \sinh(t) \bigg),
\\
\widetilde{u}_{(5/2)}(t) & = \frac{1}{192}  \bigg(3 t^2 (t^2-1) \cosh(t) - t(  2 t^2-3) \sinh(t) \bigg).
\end{align}
For general $k,$ we have 
\begin{equation}
\widetilde{u}_{(k+1/2)}(t)=   \widehat{P}_{k,1}(t)\,t\, \cosh(t) + \widehat{P}_{k,2}(t) \,\sinh(t), 
\end{equation}
where $\widehat{P}_{k,1}$ and $\widehat{P}_{k,2}$ are certain odd polynomials of degree $k+1$, from which it follows 
\begin{equation}\label{3.20a}
{r}_{(k+1/2)}(x)= \frac{1}{2}  \bigg (  \widehat{P}_{k,1}(-\partial_x)(-\partial_x) \Big ( \delta(x-1) + \delta(x+1) \Big ) +
\widehat{P}_{k,2}(-\partial_x) \Big ( \delta(x-1) - \delta(x+1) \Big )
\bigg );
\end{equation}
cf.~(\ref{3.20}).
Notice also that  
\begin{align} \label{u 32 0}
\widetilde{u}_{(3/2)}(0)= \widetilde{u}_{(3/2)}''(0)=0, \qquad  \widetilde{u}_{(3/2)}''''(0)=1, \qquad \widetilde{u}_{(3/2)}^{(6)}(0)=3, \qquad \widetilde{u}_{(3/2)}^{(8)}(0)=6, \qquad \widetilde{u}_{(3/2)}^{(10)}(0)=10.
\end{align}
These coincide with the coefficients of the ${\rm O}(1/N)$ term in \eqref{terminating p2345}. 
Along the same lines, we have
\begin{align}  \label{u 52 0}
\widetilde{u}_{(5/2)}(0)= \widetilde{u}_{(5/2)}''(0)= \widetilde{u}_{(5/2)}''''(0)=0, \qquad \widetilde{u}_{(5/2)}^{(6)}(0)=4, \qquad \widetilde{u}_{(5/2)}^{(8)}(0)=22, \qquad \widetilde{u}_{(5/2)}^{(10)}(0)=70
\end{align}
which also correspond to the coefficients of the ${\rm O}(1/N^2)$ term in \eqref{terminating p2345}. To be consistent with the fact that the final sum in (\ref{1.21}) terminates, for general $k$ we must have ${\partial_t^{2j} } \tilde{u}_{(k+1/2)}(t) |_{t=0}=0$ for $j=0,\dots,k$.

\subsection{Asymptotic expansion of the Stieltjes transform}

Let us write 
\begin{equation}
\widetilde{W}(t):= \int_\R \frac{ \rho_N^{ \rm r }(\sqrt{N}x) }{ t-x }\,dx= W(\sqrt{N}t) 
\end{equation}
for the Stieltjes transform of the rescaled density. 
Then \eqref{1.33} can be rewritten as 
\begin{equation} 
\bigg( \widehat{\mathcal{D}}_0[t] + \frac{\widehat{\mathcal{D}}_1[t]}{N}+ \frac{ \widehat{\mathcal{D}}_2[t]}{N^2} \bigg) \widetilde{W}(t)= \Big(  4  - 2 t^2 + \frac{1}{N}  \Big) \frac{ M_{0,N}^{\rm r} }{ N^{1/2} } - 6 \frac{ M_{N,2}^{\rm r} }{ N^{3/2} },  
\end{equation}
where 
\begin{align}
&\widehat{\mathcal{D}}_0[t]:= 2(t^2-1)^2 \partial_t,
\\
& \widehat{\mathcal{D}}_1[t]:= (t^2-1)  (  3t\,\partial_t^2+5 \,\partial_t   ),
\\
& \widehat{\mathcal{D}}_2[t]:= t^2 \, \partial_t^3+4 t\,\partial_t^2+2 \,\partial_t . 
\end{align}
On the other hand, by Theorem~\ref{Thm_large N},  
\begin{equation}
 \Big(  4  - 2 t^2 + \frac{1}{N}  \Big) \frac{ M_{0,N}^{\rm r} }{ N^{1/2} } - 6 \frac{ M_{N,2}^{\rm r} }{ N^{3/2} }  \sim \sqrt{ \frac{2}{\pi} }  \bigg( 2-2t^2 + \sum_{l=1}^{\infty} \frac{ (4-2t^2)a_l+ a_{l-1} -6 b_l }{N^l} \bigg) - \frac{t^2+1}{ N^{1/2} } +\frac{1}{2N^{3/2}} . 
\end{equation}
Then as before, by recursively solving this system of differential equations with the initial condition $\widetilde{W}(t)={\rm O}(1/t)$ as $t \to \infty$, one can derive the expansion
\begin{equation}
\widetilde{W}(t)= \sum_{k=0}^\infty \bigg ( \frac{ \widetilde{W}_{ (k) }(t) }{ N^k } + \frac{ \widetilde{W}_{ (k+1/2) }(t) }{ N^{k+1/2} } \bigg ). 
\end{equation}
For instance, we have
\begin{align}
\widetilde{W}_{ (0) }(t) = \frac{1}{ \sqrt{2\pi} } \log \Big( \frac{t+1}{t-1} \Big), \qquad  \widetilde{W}_{ (1/2) }(t) =  \frac{t}{2(t^2-1)}, 
\end{align}
which are consistent with \eqref{r0 r12}, and  
\begin{equation}
\widetilde{W}_{ (1) }(t) = -\frac{1}{ \sqrt{2\pi} } \frac{3t(t^2-3)}{ 4(t^2-1)^2 }, \qquad   \widetilde{W}_{ (3/2) }(t)=  \frac{t}{(t^2-1)^3}. 
\end{equation}
In particular, one reads off that as $t \to \infty$, $\widetilde{W}_{(1)}(t) \asymp  t^{-1}$, whereas $\widetilde{W}_{(3/2)}(t) \asymp t^{-3}$. Generally it is required that for $t \to \infty$,
$\widetilde{W}_{(k)}(t) \asymp  t^{-1}$, whereas $\widetilde{W}_{(2k+1)/2}(t) \asymp {t^{-2k-1}}$, so as to be consistent with (\ref{1.21}).

\bibliographystyle{abbrv}

\begin{thebibliography}{100}








\bibitem{ABES23} G.~Akemann, S.-S. Byun, M. Ebke and G. Schehr, \emph{Universality in the number variance and counting statistics of the real and symplectic Ginibre ensemble},  J. Phys. A \textbf{56} (2023), 495202.






\bibitem{ABGS21}
T. Assiotis, B. Bedert, M. Gunes and A. Soor, \emph{Moments of generalized Cauchy random matrices and continuous-Hahn polynomials}, Nonlinearity, \textbf{34} (2021), 4923.


\bibitem{BIPZ78} E. Br\'{e}zin, C. Itzykson, G. Parisi and J. B. Zuber, \emph{Planar diagrams}, Comm. Math. Phys. \textbf{59} (1978), 35--51.




\bibitem{By23b} S.-S. Byun, \emph{Harer-Zagier type recursion formula for the elliptic GinOE}, arXiv:2309.11185.  



\bibitem{BF22} S.-S.~Byun and P. J.~Forrester, \emph{Progress on the study of the Ginibre ensembles I: GinUE}, arXiv:2211.16223.

\bibitem{BF23} S.-S.~Byun and P. J.~Forrester, \emph{Progress on the study of the Ginibre ensembles II: GinOE and GinSE}, arXiv:2301.05022.

\bibitem{BKLL23} S.-S. Byun, N.-G. Kang, J. O. Lee and J. Lee, \emph{Real eigenvalues of elliptic random matrices}, Int. Math. Res. Not. \textbf{2023} (2023), 2243--2280.







\bibitem{CCO20} P. Cohen, F. D. Cunden and N. O'Connell, \emph{Moments of discrete orthogonal polynomial ensembles}, Electron. J. Probab. \textbf{25} (2020), 1--19.



\bibitem{CMOS19} F. D. Cunden, F. Mezzadri, N. O'Connell and N. Simm, \emph{Moments of random matrices and hypergeometric orthogonal polynomials}, Comm. Math. Phys. \textbf{369} (2019), 1091--1145.





\bibitem{Ed97} A.~Edelman, \emph{The probability that a random real {G}aussian matrix has $k$ real eigenvalues, related distributions, and the circular law}, J. Multivariate. Anal. \textbf{60} (1997), 203--232.

\bibitem{EKS94} A. Edelman, E. Kostlan and M. Shub, \emph{How many eigenvalues of a random matrix are real?} J. Amer. Math. Soc. \textbf{7} (1994), 247--267.








\bibitem{Fo10} P. J. Forrester, \emph{Log-gases and random matrices}, Princeton University Press, Princeton, NJ, 2010.


\bibitem{Fo21} P. J. Forrester, \emph{Moments of the ground state density for the $d$-dimensional Fermi gas in an harmonic trap}, Random Matrices Theory Appl. \textbf{10} (2021), no. 2, Paper No. 2150018, 18 pp.




\bibitem{FI16} P. J. Forrester and J. R. Ipsen, \emph{Real eigenvalue statistics for products of asymmetric real Gaussian matrices}, Linear Algebra Appl. \textbf{510} (2016), 259--290.



\bibitem{FN07} P. J. Forrester and T. Nagao, \emph{Eigenvalue statistics of the real Ginibre ensemble}, Phys. Rev. Lett. \textbf{99} (2007), 050603.


\bibitem{FR09} P. J. Forrester and E. Rains, \emph{Matrix averages relating to Ginibre ensembles}, J. Phys. A \textbf{42} (2009), 385205.

\bibitem{FLSY23} P. J. Forrester, S.-H. Li, B.-J. Shen and G.-F. Yu, \emph{$q$-Pearson pair and moments in $q$-deformed ensembles}, Ramanujan J. \textbf{60} (2023), 195--235.

\bibitem{FR21} P. J. Forrester and A. Rahman, \emph{Relations between moments for the Jacobi and Cauchy random matrix ensembles}, J. Math. Phys. 62 (2021), 073302.


 





\bibitem{GGR21} M. Gisonni, T. Grava and G. Ruzza, \emph{Jacobi ensemble, Hurwitz numbers and Wilson polynomials}, Lett. Math. Phys. \textbf{111} (2021), no. 3, Paper No. 67, 38 pp. 


\bibitem{GJ97} I. Goulden and D. Jackson, \emph{Maps in locally orientable surfaces and integrals over real symmetric surfaces}, Can. J. Math. \textbf{49} (1997) 865--882.



\bibitem{HT03} U. Haagerup and S. Thorbjørnsen, \emph{Random matrices with complex Gaussian entries}, Expo. Math. \textbf{21} (2003), 293--337.

\bibitem{HZ86} J. Harer and D. Zagier, \emph{The Euler characteristic of the moduli space of curves}, Invent. Math. \textbf{85} (1986), 457--485.





\bibitem{Ku15} S.~Kumar, \emph{Exact evaluations of some {M}eijer {G}-functions and probability of all eigenvalues real for products of two {G}aussian matrices}, J. Phys. A \textbf{48} (2015) 445206.

\bibitem{Le09} M. Ledoux, \emph{A recursion formula for the moments of the Gaussian orthogonal ensemble}, Ann. Inst. Henri Poincaré Probab. Stat. \textbf{45} (2009), 754--769.





\bibitem{MS11} F. Mezzadri and N. Simm, \emph{Moments of the transmission eigenvalues, proper delay times, and random matrix theory I}, J. Math. Phys. \textbf{52} (2011), 103511.





\bibitem{NIST}  F. W. J. Olver, D. W. Lozier, R. F. Boisvert and C. W. Clark, (eds): \emph{NIST Handbook of Mathematical Functions}, Cambridge: Cambridge University Press, 2010.



\bibitem{RF21} A. Rahman and P. J. Forrester, \emph{Linear differential equations for the resolvents of the classical matrix ensembles}, Random Matrices Theory Appl. \textbf{10} (2021), no. 3, Paper No. 2250003, 43 pp.








\bibitem{SK09} H.-J. Sommers and B. A. Khoruzhenko, \emph{Schur function averages for the real Ginibre ensemble}, J. Phys. A \textbf{42} (2009), 222002. 


\bibitem{SW23} A.~Soshnikov and C.~Wu, \emph{A note on cumulant technique in random matrix theory}, Entropy \textbf{23} (2023), 725.


\bibitem{Wi55} E. P.~Wigner, \emph{Characteristic vectors of bordered matrices with infinite dimensions}, Ann. Math. \textbf{62} (1955), 548--564.  
  
\bibitem{Wi58} E. P.~Wigner, \emph{On the distribution of the roots of certain symmetric matrices}, Ann. Math. \textbf{67} (1958), 325--327.  

\bibitem{WF14} N. Witte and P. J. Forrester, \emph{Moments of the Gaussian $\beta$ ensembles and the large $N$ expansion of the densities}, J. Math. Phys. \textbf{55} (2014), 083302. 


 \end{thebibliography}

\end{document}